%% file: main.tex
\documentclass[acmtog,nonacm,balance=false]{acmart}
\AtBeginDocument{%
  }

\setcopyright{acmlicensed}
\copyrightyear{2018}
\acmYear{2018}
\acmDOI{XXXXXXX.XXXXXXX}
\acmConference[Conference acronym 'XX]{Make sure to enter the correct
  conference title from your rights confirmation email}{June 03--05,
  2018}{Woodstock, NY}
\acmISBN{978-1-4503-XXXX-X/2018/06}
\let\oldnl\nl%
\newcommand{\nonl}{\renewcommand{\nl}{\let\nl\oldnl}}%
\usepackage{algorithm}
\usepackage[noend]{algpseudocode}
\newlength\savedwidth

\usepackage{enumitem} 
\usepackage{color, colortbl}

\usepackage{soul}
\usepackage{multirow}
\usepackage{bbding}

\definecolor{headercolor}{HTML}{b2df8a}
\definecolor{header2base}{HTML}{AE2573}
\colorlet{header2color}{header2base!30}  %

\input{newcommands}

\begin{document}

\title{SpUDD: \textbf{S}uper\textbf{p}ower Contouring of \textbf{U}nsigned \textbf{D}istance \textbf{D}ata}

\author{Ningna Wang}
\affiliation{%
  \institution{Columbia University}
  \state{New York}
  \country{USA}
}
\email{ningna.wang@columbia.edu}

\author{Xiana Carrera}
\affiliation{%
  \institution{Columbia University}
  \state{New York}
  \country{USA}
}
\email{x.carrera@columbia.edu}

\author{Christopher Batty}
\affiliation{%
  \institution{University of Waterloo}
  \city{Waterloo}
  \country{Canada}
}
\email{christopher.batty@uwaterloo.ca}

\author{Oded Stein}
\affiliation{%
  \institution{Technion}
  \city{Haifa}
  \country{Israel}
}

\affiliation{%
  \institution{University of Southern California}
  \city{Los Angeles}
  \state{CA}
  \country{USA}
}
\email{ostein@usc.edu}

\author{Silvia Sellán}
\affiliation{%
  \institution{Columbia University}
  \city{New York}
  \country{USA}
}
\email{silviasellan@cs.columbia.edu}

\renewcommand\shortauthors{Wang et. al}   %

\begin{teaserfigure}
\vspace{-0.25cm}
  \includegraphics[width=\textwidth]{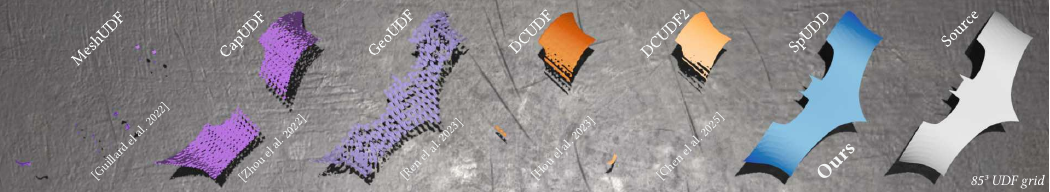}
  \vspace{-0.7cm}
  \caption{Prior methods focus on reconstructing meshes from either \emph{continuous} unsigned distance functions or a discrete set of \emph{signed}  distance samples, but fail if restricted to an input that is \emph{both} unsigned \emph{and} discrete (left). We theoretically study this more challenging reconstruction task and design the first method tailor-made to generate polygonal meshes from discrete unsigned distance data (right).}
  \label{fig:teaser}
\end{teaserfigure}

\begin{abstract}

Unsigned distance functions offer a powerful and flexible implicit surface representation that, unlike their signed counterparts, allow for surfaces that are open, non-orientable, or non-manifold.
We consider the problem of reconstructing arbitrary surfaces from a finite set of samples of unsigned distance data. 
Existing methods for mesh reconstruction from distance data rely on sign information, accurate gradients, a corresponding continuous distance function, or extensive data-dependent training. However, they fail when applied to input that is both discrete and unsigned.
Inspired by this challenge, we study the power diagram generated by the distance samples and propose a novel theoretical concept, the \emph{superpower contour}, which we prove converges to the true surface in the limit of sampling density.
We use this superpower contour as an initial surface proxy and design an algorithm that leverages it to produce a polygonal mesh  approximating the unknown true geometry.
Our method vastly outperforms other conceivable strategies for the discrete unsigned distance reconstruction task, and sets the stage for future work on this mathematically rich problem. 
\end{abstract}

\maketitle

\input{sections/1_intro.tex}

\input{sections/2_relatedwork.tex}
\input{sections/3_method.tex}
\input{sections/4_experiment.tex}

\input{sections/5_conclusion}

\begin{acks}
The Geometry and the City lab at Columbia University is supported by generous gifts from nTop, Adobe Inc., Dandy, Braid Technologies, and a Supported Research Project by Dreamsports. Christopher Batty is supported by the Natural Sciences and Engineering Research Council of Canada (Grant RGPIN-2021-02524). Oded Stein is generously supported by a gift from Adobe Inc. and the National Science Foundation (Award \#2335493).
\end{acks}

\bibliographystyle{ACM-Reference-Format}
\bibliography{reference}

\appendix

\begin{figure*}[h]
\centering
\includegraphics[width=\linewidth]{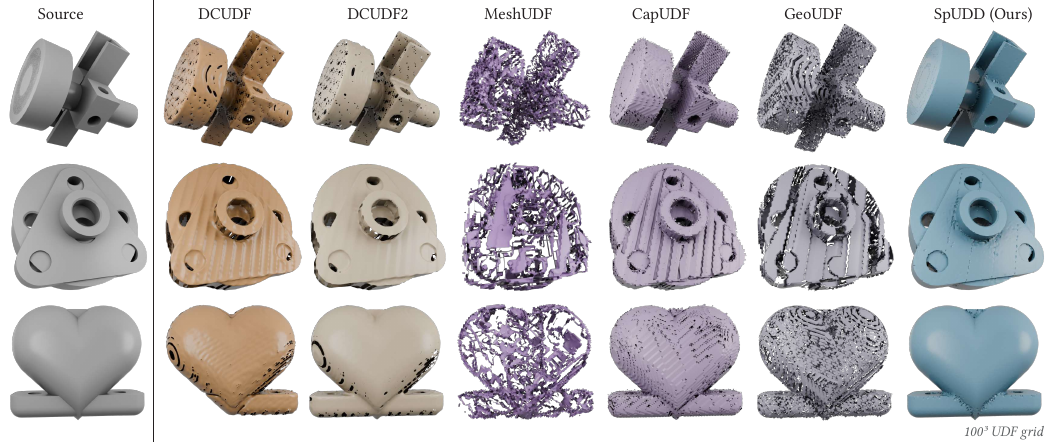}
\caption{Visual comparison with 5 state-of-the-art methods using $100^3$ UDF grid on the SALS dataset~\cite{ren2025sals}.}
\label{fig:dataset_SALS}
\end{figure*}

\section{Lemma in proof of \autoref{thm:convergence}}
\label{app:lemma}

\begin{lemma}\label{lem:union-convergence}
    In the conditions of \autoref{thm:convergence}, let $(\point_1,\distance_1),\dots,(\point_n,\distance_n)$ be independently drawn samples of distance to an unknown surface $\surface$, and let $\mathcal{U}_n$ be the complement of the union of all unsigned distance balls,
    \begin{equation}
        \mathcal{U}_n := \compactset \setminus \bigcup_{i=1}^n \ball(\point_i,\distance_i).
    \end{equation}
    Then,
    \begin{equation}
        \mathcal{U}_n \xrightarrow[n\to\infty]{d_H} \surface
    \end{equation}
    where $\xrightarrow[n\to\infty]{d_H}$ denotes convergence in two-sided Hausdorff distance.
\end{lemma}
\begin{proof}
    Let $x$ be in $\surface$, and assume $x\notin \mathcal{U}_n$. Then, there must exist a point $\point_i$ such that $x \in \ball(\point_i,\distance_i)$. But then that means $\|\point_i-x\| < \distance_i$, which since $x \in \surface$, contradicts the definition of distance. Therefore, $x\in \mathcal{U}_n$ and the one-sided Hausdorff distance from $\surface$ to $\mathcal{U}_n$ is always zero.

    For the other direction, let $x \in \mathcal{U}_n \subset \compactset$ and let
  $\delta = \dist(x, \surface) > 0$.  Since $\distr$'s support contains
  $\compactset$ and $x\in\compactset$, for large enough $n$ there exists a sample $\point_i$
  with $\|\point_i - x\| < \delta/2$.  By the triangle inequality,
  \begin{equation}
      \distance_i = \dist(\point_i, \surface) \geq \dist(x,\surface) -
      \|\point_i - x\| > \delta - \frac{\delta}{2} = \frac{\delta}{2}
      > \|\point_i - x\|,
  \end{equation}

  so $x \in \ball(\point_i,\distance_i)$, contradicting $x \in
  \mathcal{U}_n$.  Therefore no point of $\compactset$ at positive distance from
  $\surface$ can lie in $\mathcal{U}_n$ for large $n$, and the one-sided Hausdorff
  distance from $\mathcal{U}_n$ to $\surface$ goes to zero.
\end{proof}

\end{document}

%% file: newcommands.tex
\theoremstyle{definition}
\newtheorem{definition}{Definition}
\newtheorem{theorem}{Theorem}

\newtheorem{lemma}[theorem]{Lemma}

\definecolor{algshade}{gray}{0.5}
\newcommand{\graystate}[1]{\State \textcolor{algshade}{#1}}

\makeatletter
\algdef{SE}[FOR]{GrayFor}{GrayEndFor}[1]%
  {\textcolor{algshade}{\textbf{for} #1 \textbf{do}}}%
  {\textcolor{algshade}{\textbf{end for}}}

\algdef{SE}[IF]{GrayIf}{GrayEndIf}[1]%
  {\textcolor{algshade}{\textbf{if} #1 \textbf{then}}}%
  {\textcolor{algshade}{\textbf{end if}}}

\newcommand{\graylinefrom}{0}
\newcommand{\graylineto}{0}
\algrenewcommand\alglinenumber[1]{%
  \ifnum#1<\graylinefrom\relax\footnotesize#1:\else
    \ifnum#1>\graylineto\relax\footnotesize#1:\else
      \textcolor{algshade}{\footnotesize#1:}\fi\fi}
\makeatother

\newcommand{\cellvert}{\mathbf{x}}

\newcommand{\centroid}{\mathbf{c}}

\newcommand{\energy}{E}
\newcommand{\dist}{\text{dist}}

\newcommand{\distance}{d}

\newcommand{\mesh}{\mathcal{M}}
\newcommand{\bigO}{\mathcal{O}}
\newcommand{\point}{p}
\newcommand{\closestpoint}{b}
\newcommand{\sdistance}{s}
\newcommand{\hpoint}{q}
\newcommand{\hnormal}{\vec{n}}
\newcommand{\Hermitepoint}{\hpoint}
\newcommand{\Hermitenormal}{\hnormal}
\newcommand{\cell}{c}
\newcommand{\vertex}{v}
\newcommand{\edge}{e}
\newcommand{\distr}{\mathcal{D}}
\newcommand{\compactset}{\mathcal{K}}
\newcommand{\surface}{\Omega}
\newcommand{\R}{\mathbb{R}}
\newcommand{\pd}{\pi}
\newcommand{\pc}{\mathcal{P}}
\newcommand{\spc}{\mathcal{SP}}
\newcommand{\ball}{\mathcal{B}}
\newcommand{\pdc}{\mathcal{C}}
\newcommand{\pdf}{\mathcal{F}}

\newcommand{\hd}{\textit{HE}}
\newcommand{\cd}{\textit{CE}}
\newcommand{\ecd}{\textit{ECE}}

\newcommand{\best}[1]{\textbf{#1}}
\newcommand{\second}[1]{\underline{#1}}

\newcommand{\eg}{e.g., }

\usepackage{layouts}
\usepackage{calc}
\makeatletter
\newcommand{\layoutdetails}{%
\begin{tabular}{ll}
\texttt{\textbackslash{textwidth}} & \printinunitsof{in}\prntlen{\textwidth} \\
\texttt{\textbackslash{linewidth}} & \printinunitsof{in}\prntlen{\linewidth} \\
Main text font &  \f@size pt \f@family \\
\sffamily \small Caption text font &  \sffamily \small \f@size pt \f@family \\
\end{tabular}%
}
\makeatother

%% file: sections/1_intro.tex
\vspace{-0.15cm}
\section{Introduction}
\label{sec:intro}

\emph{Unsigned Distance Functions} (UDFs) implicitly represent shapes by measuring the absolute distance from any position in space to their surfaces.
Their expressiveness and flexibility have made them a powerful representation in 3D modeling: in contrast to \emph{signed} distance fields, they can represent shapes with arbitrary topology, including non-orientable and open surfaces as well as surfaces with non-manifold features.
From classical tools to 3D generative AI, many modeling workflows rely on first sampling or generating a discrete set of UDF data, before performing downstream tasks that require explicit surface representations (\eg texture mapping, physical simulation, and animation). Consequently, \emph{reconstructing meshes from only discrete UDF samples} is critically important.

Surprisingly, to the best of our knowledge, this specific task has not been considered in the geometry processing community before.
General UDF-to-mesh reconstruction methods assume additional input data, be it a continuous function returning the UDF value at arbitrary spatial positions \cite{zhang2023dualmeshUDF}, the UDF's gradients \cite{ren2023geoudf,zhou2022capudf,guillard2022meshudf}, or large datasets on which to train neural strategies \cite{chen2022neuraldc,stella2024nsdudf,maruani2024ponq}.
When restricted to only a finite set of discrete UDF samples (\eg by interpolating UDF values from a grid), these methods generally fail, often catastrophically (see \autoref{fig:teaser}).
Recently developed algorithms for meshing discrete \emph{signed} distance are promising, but make heavy use of the function's sign to spatially localize the surface \cite{mnmeg,dcsdd}, decide its topology \cite{rfts}, and estimate its normal direction \cite{mnmsiggraph,rfta}.
Naively extending them to unsigned inputs results in prominent, uncontrollable artifacts (see \autoref{fig:teaser}).

In this paper, we introduce the first method to reconstruct shapes of general topology from only discrete UDF data, without assuming any additional input information.
Our key insight is twofold: first, we note that the contouring proposed by \citet{dcsdd} is mostly agnostic to the distance data's sign, relying on it only in the beginning of the algorithm to decide which grid edges are crossed by the surface.
Secondly, we draw inspiration from the \emph{power contour} proposed by \citet{mnmeg}, which provides valuable information about the surface location but is only defined for \emph{signed} distance inputs. We show that it can be generalized to unsigned data by proposing a related larger structure, which we call the \emph{superpower contour}.
A theoretical study of this contour reveals several useful properties: most critically, it converges to the true unknown surface as sampling density increases, justifying its use as a proxy during reconstruction.
Combining these two insights, our algorithm relies on the superpower contour to identify the active edges of a background grid and the topology of the new mesh, and uses a dual optimization strategy to decide on the latter's geometry.

\begin{figure*}[h]
\vspace{-0.4cm}
\centering
\includegraphics{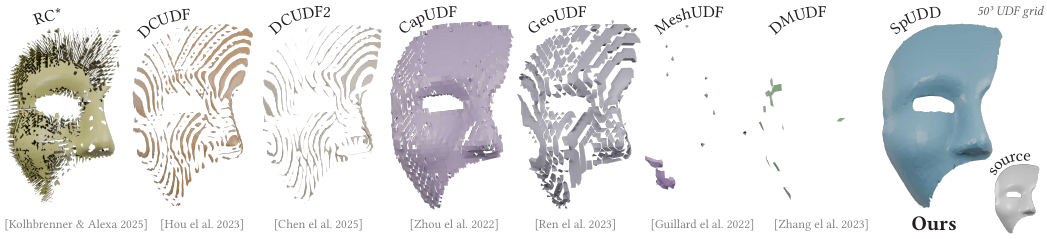}
\vspace{-0.4cm}
\caption{Our method (right) manages to reconstruct surfaces from coarse, \emph{discrete} unsigned distance data while avoiding the artifacts of prior work.}
\vspace{-0.3cm}
\label{fig:baselines_simple_shape}
\end{figure*}

Both computing our superpower contour and adapting the algorithm proposed by \citet{dcsdd} to accept unsigned distances as input require nuanced algorithmic decisions, which we justify and ablate experimentally in this text.
Critically, we reduce the inherent ambiguity in unsigned distance data without limiting the class of reconstructed surfaces: our algorithm can produce complex open, non-orientable, and non-manifold surfaces as easily as closed, oriented, manifold ones (see \autoref{fig:teaser}).

Through qualitative and quantitative tests on three different benchmarks, we demonstrate the failure of existing algorithms when applied to this challenging, discrete reconstruction task.
With these, we hope to introduce an unsolved, timely and mathematically rich research question to the geometry processing community.
At the same time, we evaluate our method's suitability for this task on three different large-scale datasets of diverse origins and geometric characteristics, highlighting its impact with prototypical 3D modeling applications and setting the stage for future work.

%% file: sections/2_relatedwork.tex
\section{Related Work}
\label{sec:related}

The key benefit of unsigned distance functions is their increased representational flexibility compared to their signed counterparts. This property has led to recent interest in continuous (neural) UDFs in the context of deep learning (e.g., \cite{chibane2020neural}), but their discrete counterparts have long been studied in several communities. In 2D image processing and analysis, Euclidean Distance Transforms (EDTs) are  algorithms to compute exact grid-based UDFs. \citet{fabbri20082d} list a host of downstream applications, including segmentation, morphological operations (such as erosion and offsetting), computation of various geometric objects (such as skeletonization, generalized Voronoi diagrams, and medial axes), pathfinding and collision avoidance for robotics, and shape matching. 

Naturally, many of these same operations are relevant to 3D problems in graphics, geometry processing, and medical imaging, among others; \citet{jones20063d} offers a survey of discrete distance fields (signed and unsigned) in graphics, and their applications. In applied mathematics, level set methods \cite{sethian1999level} for evolving surfaces more commonly employ discrete SDFs, but variants that involve discrete UDFs have been developed for multiphase flow (i.e., non-manifold interfaces) \cite{zheng2009simulation,saye2011voronoi} and for evolving open curves and surfaces \cite{leung2009grid}. The closely related (discrete) closest point transform has been used in the Closest Point Method for solving PDEs on general surfaces \cite{macdonald2008level,king2024closest}.
As a practical application, we will demonstrate in \autoref{sec:applications} the ability to perform union(-like) operations on discrete UDFs for implicit geometric modeling with open and non-manifold surfaces.

Motivated by the above, we consider the problem of reconstructing meshes from a discrete set of unsigned distance samples.
While we have not been able to find any prior work dedicated to this specific problem, it is closely related to two parallel, active lines of research: mesh reconstruction from \emph{continuous} UDFs, and mesh reconstruction from discrete \emph{signed} distance data.
We will review the state of the art for these two below, highlighting the inability of the first methods to handle discrete data and drawing the key insights from the second that lay the groundwork for our algorithm.

\begin{figure*}[t]
\centering
\vspace{-0.4cm}
\includegraphics{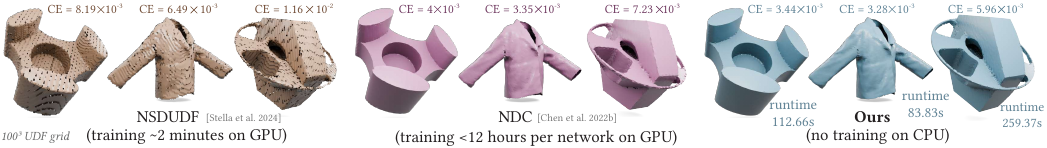}
\vspace{-0.3cm}
\caption{Our method quantitatively (using Chamfer error $\cd$) matches or outperforms even the best-performing neural methods in UDF reconstruction.}
\vspace{-0.3cm}
\label{fig:neural}
\end{figure*}

\subsection{Reconstruction from continuous UDFs}
\label{sec:related-udfs}

Broadly speaking, virtually all strategies for meshing continuous UDFs consist of two steps. First, a spatial subdivision is used to \emph{localize} the volumetric regions containing the surface (a far-from-trivial task in the case of \emph{unsigned} distances).
Then, variations of classical methods from the reconstruction of \emph{signed} distance functions are used to produce local surface patches that are joined together to create a mesh.
It is thus useful to categorize prior work in this area based on the specific method chosen for the second of these steps.

First, recent adaptations of classic \emph{marching cubes} assume knowledge of the UDF \emph{and its gradient} at the vertices of a regular grid \cite{ren2023geoudf,zhou2022capudf,guillard2022meshudf}. (In a similar approach, \citet{congote2010marching} queries extra edge midpoint samples to fit a parabola per edge.) Variations in gradient direction within each cell are then used to establish a local \emph{pseudo-sign} at each vertex, enabling reconstruction via the traditional Marching Cubes lookup tables \cite{marchingcube,wyvill1986data}.
Unfortunately, the reliance on the surface's gradient cannot be easily circumvented: as shown in \autoref{fig:baselines_simple_shape} (center, purple hues), estimating this gradient only from the grid's UDF samples (\eg with central finite differences) results in topological noise, holes, and artifact-laden (often fully empty) outputs. 

Second, \emph{dual contouring} methods (\eg \cite{zhang2023dualmeshUDF}) densely sample the continuous UDF to guide the construction of an adaptive spatial subdivision structure like an octree. In each leaf node, further sampling of the UDF is used to define a quadratic error function, whose minimization produces the mesh's local geometry in a way similar to the original Dual Contouring work by \citet{ju2002dual}.
While impressive at reconstructing sharp features, even from noisy inputs, such methods depend heavily on the ability to query the UDF at arbitrary locations. If forced to use only grid values (\eg by trilinearly interpolating from them), they produce either topological noise or completely empty outputs (see \autoref{fig:baselines_simple_shape}, green hue).

Third, \emph{double cover} methods \cite{hou2023dcudf,chen2025dcudf2} use Marching Cubes \cite{Lorensen1987,wyvill1986data} to mesh an epsilon level set of the UDF, yielding a double-layered envelope containing the desired surface.
This double layer is then evolved following a geometric flow, which relies on densely sampling the UDF function (and/or its gradient) at all arbitrary spatial positions occupied by the surface during its evolution.
A final stage cuts and stitches together sections of both layers to produce a single-layer output.
Both these flow evolution and cutting steps lean on the ability to densely query the UDF: if made to rely only on the input grid values, they show prominent topological aberrations correlated with grid alignment (see \autoref{fig:baselines_simple_shape}, center left, orange hues).

Finally, \emph{data-driven} methods utilize large datasets and computational resources to train neural networks that localize the surface and identify its local topology and geometry \cite{stella2024nsdudf,chen2022neuraldc,maruani2024ponq}.
Some of these \cite{stella2024nsdudf} then further rely on gradient information or dense UDF sampling to produce the reconstruction, making them unsuitable for our task (see \autoref{fig:neural}, left). Others, like \emph{Neural Dual Contouring} \cite{chen2022neuraldc}, are then able to leverage the data to produce accurate reconstructions from discrete UDF samples, especially for shapes similar to those in their training set (\autoref{fig:neural}, center).

From this overview of continuous UDF reconstruction, exemplified in \autoref{fig:baselines_simple_shape}, we conclude that no existing work is able to faithfully reconstruct explicit surfaces only from the little information contained in a discrete set of UDF samples, without resorting to additional information like gradients, on-demand sampling, or priors learned from large datasets.
It is this realization that motivates the need for our algorithm (see \autoref{fig:baselines_simple_shape}, right).
To develop it, we combine and adapt several insights from recent work on discrete \emph{signed} distance reconstruction, which we review below.

\vspace{-0.1cm}
\subsection{Reconstruction from discrete SDFs}

In contrast to the unsigned case, reconstructing surfaces from discrete \emph{signed} distance samples has received a large amount of attention \cite{deAraujo2015}.
Partly, this is because one can avoid the elaborate localization schemes described in \autoref{sec:related-udfs} and instead produce surface patches only on grid cells or edges along which the SDF changes sign. 
These methods can typically be categorized into \emph{spatial subdivision} (\eg Marching Cubes \cite{marchingcube,wyvill1986data}), \emph{surface tracking} \cite{hilton1996marching} and \emph{shrinkwrapping} \cite{stander1997}; a full review of this extensive research area exceeds the scope of this paper.

An important precedent related to our task is the work of \citet{mullen2010signing}, which uses UDFs as an intermediate step in point cloud reconstruction. They construct a discrete UDF from the point cloud, estimate a global \emph{signing} of the UDF data to produce an SDF, smooth the SDF to close the surface and eliminate small holes, and finally use a Delaunay refinement approach for contouring. Similarly, \citet{xu2020signed} use epsilon level sets to convert unsigned distance data from polygon soups into discrete signed distance data. However, these approaches assume that a sufficiently consistent global signing is possible, which is often challenging for complex open surfaces and impossible for general non-manifold and non-orientable ones.

Classically, SDF contouring methods treat SDFs as generic implicit functions, but recent efforts have begun explicitly exploiting their Euclidean distance property. Early work in this vein by \citet{rfts} poses the global discrete SDF reconstruction problem as that of finding a surface that is tangent to a set of spheres without intersecting them, relying on the sign of each SDF sample to impose the direction of tangency. Follow-up studies by \citet{rfta} and \citet{mnmsiggraph} use a similar characterization to compute a set of tangency points. The points' normal directions are deduced from the signs of the SDF data, resulting in an oriented point cloud that can be passed to off-the-shelf reconstruction methods \cite{Kazhdan2006,Kazhdan2013}.

A critical building block for our method is the realization, made recently by \citet{mnmeg}, that signed distance samples can be used to construct a spatial subdivision structure called the \emph{power diagram} \cite{aurenhammer1987power}, whose dual structure is known as the \emph{regular tetrahedralization}. The faces of this diagram can be filtered based on the signs of the SDF samples to obtain the \emph{power contour}; then, the dual tetrahedra associated to the power contour's vertices can be used as a background mesh on which to reconstruct the surface with Marching Tetrahedra \cite{doi1991efficient}.
Unfortunately, this power contour can only be defined for \emph{signed} distance inputs (see \autoref{fig:pc_spc_2d_3d}). In \autoref{sec:method-theory}, we will borrow some of the key insights from this work to construct a more general structure, which we call the \emph{superpower contour}, that can be defined for any general (unsigned) distance data. 

In parallel, very recent work by \citet{dcsdd} proposes a contouring strategy for signed distance data sampled on a regular grid. First, the SDF signs are used to decide on a set of \emph{active edges} and \emph{active cells}; then, an iterative optimization alternates between optimizing the local geometry within each cell for tangency and positioning the surface crossing point on each active edge.
Our work is born out of the realization that much of the optimization strategy proposed by \citet{dcsdd} is independent of the samples' signs, relying on them only to decide on the set of active cells and to initialize the surface crossing information in each active edge.
While this reliance makes the existing algorithm inapplicable to unsigned data, we make use of our novel \emph{superpower contour} to develop a sign-agnostic initialization strategy that we combine with their strategy to produce meshes from unsigned distance data.

%% file: sections/3_method.tex
\vspace{-0.1cm}
\section{The \emph{Superpower Contour}}
\label{sec:method-theory}

The input to our method is a finite set of $n$ distance samples $(\point_1,\distance_1),$ $\dots,$ $(\point_n,\distance_n)$, where $\distance_i \in\R$ is the distance from $\point_i\in\R^3$ to its closest point on an unknown surface $\surface\subset\R^3$. In this section, we develop an algorithm that outputs $\mesh$, a quad mesh that approximates $\surface$.

Inspired by prior work in \emph{continuous} UDF reconstruction (e.g., \cite{zhang2023dualmeshUDF}), we will produce $\mesh$ from our \emph{discrete} set of distance data in two stages: an initial global localization of the surface (introduced in this section) followed by a local optimization in each cell of a regular grid (introduced in \autoref{sec:method-practice}). 

Absent the inside-outside information provided by a signed distance, and without the ability to densely sample the UDF to search for its zero level set, the most challenging aspect of the discrete UDF reconstruction problem is that of localizing the surface to a specific region of space. Inspired by \citet{mnmeg}, we will treat $\{(\point_i,\distance_i)\}$ as a set of seeds, and make use of the \emph{power distance} from any spatial position $x\in\R^3$ to each of the seeds, traditionally defined as 
\begin{equation}\label{equ:pd}
 \pd_i(x) = \|x-\point_i\|^2 - \distance_i^2.
\end{equation}
Intuitively, $\pd_i$ can be interpreted as measuring the distance from $x$ to a sphere centered at $\point_i$ with radius $\distance_i$. Unlike a Euclidean distance, however, its rate of change is not constant (note the lack of a square root in \autoref{equ:pd}), and it will fall off faster the further $x$ is from $\point_i$.

\begin{figure}
\centering
\includegraphics[width=\linewidth]{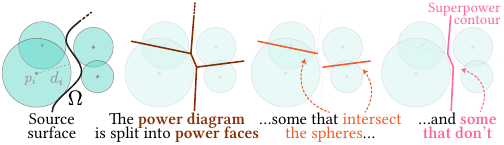}
\vspace{-0.6cm}
\caption{Intuitively, the superpower contour $\spc$ is built by filtering out the power faces that intersect any of the distance spheres. This produces a structure that serves as a proxy of the true unknown geometry $\surface$.}
\vspace{-0.3cm}
\label{fig:didactic-power-faces}
\end{figure}

\begin{figure*}
\centering
\includegraphics[width=\linewidth]{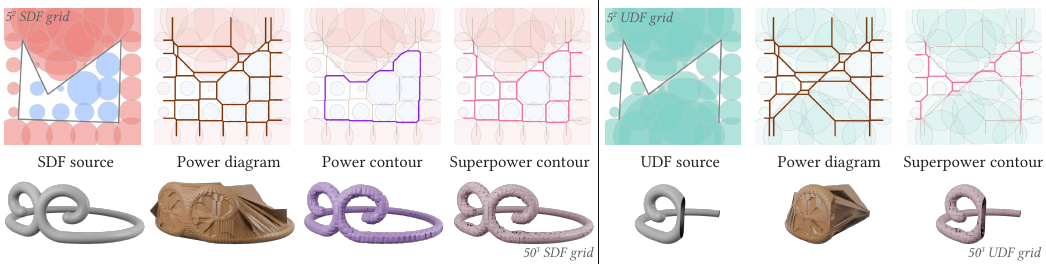}
\vspace{-0.7cm}
\caption{\citet{mnmeg} introduce the \emph{power diagram}, a structure that can be used as a proxy for the surface represented by an SDF. Our \emph{superpower contour} contains it, and is well defined for unsigned inputs (right).}
\label{fig:pc_spc_2d_3d}
\end{figure*}

In a similar way to how a Voronoi diagram partitions space based on the closest seed in Euclidean distance, the \emph{power diagram} consists of \emph{power cells} $\pdc_i$ that partition space based on the closest seed in power distance:
\begin{equation}\label{equ:pdcell}
 \pdc_i = \{x\in\R^3 : \pd_i(x) < \pd_j(x), \forall j\neq i\}.
\end{equation} 
Critically, the nonlinear falloff of the power distance makes the power cells convex, bounded by planar \emph{power faces} made up of points equi-(power-)distant to two seeds:
\begin{equation}\label{equ:pdface}
 \pdf_{i,j} = \{x\in\R^3 : \pd_i(x) = \pd_j(x) < \pd_k(x), \forall k\neq i,j\}.
\end{equation}
This property makes it possible to compute the power diagram robustly and efficiently even for a large number of seeds, using off-the-shelf libraries with exact predicates like CGAL~\cite{cgal}.
Often accompanied by its dual structure (the \emph{regular tetrahedralization}) the power diagram has been used in a myriad of geometry processing applications, from reconstruction from point clouds \cite{Ede93,Ede94} to medial axis extraction \cite{Ame01,Ame01b,2022MATFP,wang2024mattopo,wang2025matstruct}.

If the distances $\distance_i$ were signed, this sign could be used to label each power cell as \emph{inside} or \emph{outside} $\surface$. In this context, \citet{mnmeg} define the \emph{power contour} $\pc$ as the subset of power faces $\pdf_{i,j}$ that separate an inside cell from an outside one (see \autoref{fig:pc_spc_2d_3d}). They use the dual tetrahedra associated with the vertices of this contour, along with an isosurfacing strategy, to produce an approximation of the unknown surface.

In our \emph{unsigned} setting, we will instead define a new structure, which converges to the true surface $\surface$ and is guaranteed to contain the power contour introduced by \citet{mnmeg} in the specific case of closed surfaces with signed distance data.
The intuition behind it is simple, and it is shown in \autoref{fig:didactic-power-faces}: in the neighborhood of a surface, the spheres defined by the UDF data $(\point_i,\distance_i)$ will be locally separated into two non-overlapping regions of overlapping spheres. Between the two regions lie faces of the power diagram that do not intersect any of the spheres and that can serve as a good approximation of the unknown surface; away from the surface, a large number of power diagram faces intersect one of the regions and are likely irrelevant to the reconstruction.
We formalize this distinction as follows:

\begin{definition}[Superpower Contour]
Given a set of distance data $(\point_1,\distance_1),$ $\dots,$ $(\point_n,\distance_n)$, with associated power faces $\{\pdf_{i,j}\}$, the \emph{superpower contour} $\spc$ is the subset of power faces that do not intersect the interior of any of the seeds' spheres (i.e., the open balls $\ball (\point_k,\distance_k)$). In other words,
\begin{equation}\label{equ:spc}
 \spc = \bigcup_{i,j} \{\pdf_{i,j} : \pdf_{i,j} \cap \ball (\point_k,\distance_k) = \emptyset,\ \forall k\}.
\end{equation}
\end{definition}

The most critical property of the superpower contour is its ability to serve as a proxy for the unknown geometry $\surface$, enabling its use as a surface localization tool without having to resort to dense sampling, gradients, or data-driven priors. The convergence of the superpower contour to the surface is stated and proven in the following theorem, while \autoref{sec:pc-vs-spc} studies the relation between $\spc$ and the power contour introduced by \citet{mnmeg}.
These results serve as fundamental theoretical underpinnings of our method; however, a reader interested only in the algorithmic steps of our unsigned distance reconstruction may skip directly to \autoref{sec:method-practice}. 

\begin{theorem} 
\label{thm:convergence} Let $\surface$ be a two-dimensional manifold embedded in $\R^3$ and bounded by a compact set $\compactset \in \R^3$. Let $\distr$ be any distribution whose support entirely contains $\compactset$, and let $\spc_n$ be the superpower contour defined by samples $\point_1,\dots,\point_n$ independently drawn from $\distr$ along with their distances $\distance_i$ to $\surface$.
    Then, 
    \begin{equation}
        \spc_n \xrightarrow[n\to\infty]{d_H} \surface
    \end{equation}
    where $\xrightarrow[n\to\infty]{d_H}$ denotes convergence in two-sided Hausdorff distance.
\end{theorem}

\begin{proof}
Let $x$ be any point on the superpower contour $\spc_n$. By definition, $x$ is not contained in any $\ball(\point_i,\distance_i)$ for $i=1,\dots,n$; therefore, it is in the complement of the union of these balls, which we denote as 
\begin{equation}
    x \in \mathcal{U}_n := \compactset \setminus \bigcup_{i=1}^n \ball(\point_i,\distance_i).
\end{equation}
$\mathcal{U}_n$ can be interpreted as the region of space that is not covered by any of the spheres defined by the distance samples, which converges to $\surface$ in two-sided Hausdorff distance as $n\to\infty$ (see \autoref{lem:union-convergence} in the \autoref{app:lemma}). Since $x\in\mathcal{U}_n$, this implies that
\begin{equation}
    \mathrm{dist}(x,\surface) \leq d_H(\mathcal{U}_n,\surface) \xrightarrow[n\to\infty]{} 0.
\end{equation}
and thus the one-sided Hausdorff distance from $\spc_n$ to $\surface$ goes to zero as $n\to\infty$.

Conversely, let $y$ be any point on $\surface$. Since $\surface$ is a manifold, there is a sufficiently small $\epsilon_1$ around $y$ such that $\surface \cap \ball(y,\epsilon_1)$ is homeomorphic to a disk and separates $\ball(y,\epsilon_1)$ into two disjoint connected components. We refer to these as the \emph{left} and \emph{right} components, and their separation provides a locally consistent inside-outside segmentation, since a global segmentation is unavailable in the unsigned setting. Any sphere defined by the distance data $(\point_i,\distance_i)$ that intersects $\ball(y,\epsilon_1)$ must intersect either the left or the right component, but not both (otherwise, it would intersect $\surface$ and $\distance_i$ would violate the definition of distance).
Now, let $\epsilon_2$ be the radius of the smallest ball around $y$ that intersects at least one sphere in the left component and at least one sphere in the right component (as samples become denser, one can safely assume $\epsilon_2<\epsilon_1$). 

Consider then the form of the power diagram $\pd$ in this local neighborhood. Because there is at least one sphere in the left component and at least one sphere in the right component, there must be at least one power face $\pdf_{i,j}$ that intersects $\ball(y,\epsilon_2)$ and is equi-(power) distant to a seed in the left component and a seed in the right component. This power face cannot intersect any of the spheres, as they are all contained in either the left or the right component, and thus $\pdf_{i,j} \in \spc_n$. Therefore, there is a point $z\in\spc_n$ such that $\|y-z\| \leq \epsilon_2$, meaning that
\begin{equation}
    \mathrm{dist}(y,\spc_n) \leq \epsilon_2 \, .
\end{equation}
Since $\epsilon_2$ becomes smaller and smaller as the sampling density increases, this means that the one-sided Hausdorff distance from $\surface$ to $\spc_n$ also goes to zero as $n\to\infty$.
\end{proof}

\subsection{Relation between power and superpower contours}
\label{sec:pc-vs-spc}

Critically, $\spc$ is defined without the need to resort to signed distances or prior inside-outside segmentations of the distance data. This means that we can compute it not only when these signs are costly or unavailable, but also when they are mathematically non-existent, such as for open surfaces, non-orientable surfaces, or surfaces with complex non-manifold features.
Nonetheless, for the specific case of closed surfaces, $\spc$ contains the power contour $\pc$ defined by \citet{mnmeg}, as we formalize below.

\begin{theorem}
Given a set of \emph{signed} distance data $(\point_1,\sdistance_1),$ $\dots,$ $(\point_n,\sdistance_n)$ to a closed surface $\surface$, their power contour $\pc$ is a subset of the superpower contour $\spc$ defined by the \emph{unsigned} distance data $(\point_1,|\sdistance_1|),$ $\dots,$ $(\point_n,|\sdistance_n|)$.
\label{thm:pc-in-spc}
\end{theorem}
\begin{proof}
    Let $\pdf_{i,j}$ be a power face in $\pc$ corresponding to two seeds $\point_i$ and $\point_j$ with differently signed distances $\sdistance_i\sdistance_j<0$.
    Because their signs are different, the spheres $\partial\ball(\point_i,|\sdistance_i|)$ and $\partial\ball(\point_j,|\sdistance_j|)$ are disjoint \cite{rfts}.

    Let $x$ be any point on $\pdf_{i,j}$. Because $\pdf_{i,j}$ is a power face, we have that $\pd_i(x) = \pd_j(x)$. If the value of $\pd_i(x)$ (and thus also of $\pd_j(x)$) were negative, $x$ would be contained in both spheres, contradicting the fact that they are disjoint; therefore, $\pd_i(x) \ge 0$, and similarly $\pd_j(x) \ge 0$.
    This means that $x$ is not contained in either of the balls $\partial \ball(\point_i,|\sdistance_i|)$ and $\partial \ball(\point_j,|\sdistance_j|)$.
    Since, by definition of power face, $\pd_i(x) \leq \pd_k(x)$ and $\pd_j(x) \leq \pd_k(x)$ for all $k\neq i,j$, we also have that $x$ is not contained in any of the other $\partial \ball(\point_k,|\sdistance_k|)$ with $k\neq i,j$.
    Therefore, $\pdf_{i,j}$ does not intersect any of the spheres $\partial \ball(\point_k,|\sdistance_k|)$, and thus it is contained in $\spc$.
\end{proof}

\begin{figure}
\centering
\vspace{-0.4cm}
\includegraphics[width=\linewidth]{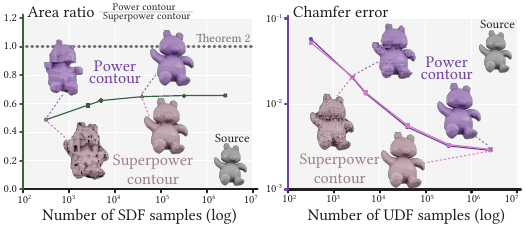}
\vspace{-0.7cm}
\caption{Like the power contour $\pc$, our superpower contour $\spc$ converges to the true surface, while maintaining a roughly constant face area ratio.}
\vspace{-0.3cm}
\label{fig:convergence_pc_spc}
\end{figure}

Importantly, together with \autoref{thm:convergence}, this result implies that the containment of $\pc$ in $\spc$ becomes tighter as the density of the distance samples increases and both $\pc$ and $\spc$ approach the true unknown surface $\surface$ (see \autoref{fig:convergence_pc_spc}).

Given this similarity between the behaviors of the superpower and power contours, one may wonder if simply swapping one for the other in the discrete SDF reconstruction algorithm proposed by \citet{mnmeg} would be sufficient to produce meshes from generic unsigned distance data.
Unfortunately, as shown in \autoref{fig:mm3_ours}, this is not the case: thanks to the SDF's sign, the power contour $\pc$ divides the vertices of every element in the regular tetrahedralization into two distinct sets (positive and negative ones), enabling Marching-Tetrahedra-like lookups to decide on the surface's local geometry.
This is no longer true of the superpower contour $\spc$, which may intersect the elements in the regular tetrahedralization in a myriad of different ways, including by splitting a tetrahedron's vertices into an odd number of connected components. This makes it impossible to assign even a locally consistent pseudo-sign in many of the tetrahedra; if one choses to simply ignore problematic tetrahedra, the resulting surfaces contain elaborate topological noise and visible aberrations (see \autoref{fig:mm3_ours}).

Instead, while the superpower contour can serve as a useful proxy for the true, unknown surface, constructing a complete discrete UDF reconstruction algorithm will require combining it with a tailor-made contouring-based meshing strategy. We detail this process in the following section.

\section{\emph{Superpower} Contouring of \emph{Unsigned} Distance Data}
\label{sec:method-practice}

As shown in \autoref{thm:convergence} and \autoref{fig:convergence_pc_spc}, the superpower contour $\spc$ can be used as a geometric proxy for the unknown surface $\surface$.
However, as a polygonal mesh, it is practically unusable for most downstream applications: it is unnecessarily large (millions of faces even at medium sampling resolutions like $100^3$ grids, see Figure~\ref{fig:pipeline}), and, moreover, the superpower contour contains abundant geometric and topological noise in the form of small non-manifold components (see \autoref{fig:pipeline}).
This motivates the need for a second step of our pipeline, in which the superpower contour is used in conjunction with a spatial subdivision strategy to output a reasonably sized quad mesh that complies with the input distance data $(\point_1,\distance_1),$ $\dots,$ $(\point_n,\distance_n)$.
In what follows, we assume that $\point_1,\dots,\point_n$ are the nodes of a regular grid bounding $\surface$. 

\subsection{Background: Dual Contouring of \emph{Signed} Distance Data}
\label{sec:signedDC}

\begin{figure}
\centering
\includegraphics[width=\linewidth]{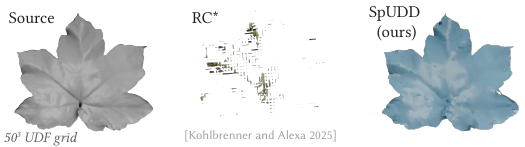}
\vspace{-0.6cm}
\caption{Merely combining our superpower contour with the Marching Tetrahedra strategy proposed by \citet{mnmeg} is not sufficient to extend their method to unsigned input (center), justifying the development of a tailor-made optimization strategy (right).}
\vspace{-0.3cm}
\label{fig:mm3_ours}
\end{figure}

We will follow a modified version of the \emph{signed} distance dual contouring strategy recently proposed by \citet{dcsdd}, itself inspired by the original work of \citet{schaefer2002dual}.
Given a regular grid of signed distance samples $(\point_1,\sdistance_1),\dots,(\point_n,\sdistance_n)$ and a list of \emph{active edges} $\edge_1, \dots, \edge_k$ (grid edges that intersect the unknown surface $\surface$), \citet{dcsdd} assume an initial set of estimated tuples $\{(\hpoint_j,\hnormal_j)\}_{i=1}^k$, i.e., an approximation of the points at which each active edge $\edge_j$ crosses $\surface$ and the surface normal at each such point.
These are known as \emph{Hermite points} and \emph{Hermite normals}.
From these, they  identify the \emph{active grid cells} (any cell containing an active edge) and assign a certain subset of the (signed) distance data to each cell.
Each active cell $\cell_i$ is endowed with a dual vertex $\vertex_i$, which is initialized to the average of the Hermite points belonging to $\cell_i$.

Then, \citet{dcsdd} pose the reconstruction problem as an energy minimization one, which is solved in a local-global fashion. In the local update, each active cell is considered individually and in parallel. For each active cell $\cell_i$ containing Hermite data $(\hpoint_1,\hnormal_1),...,(\hpoint_s,\hnormal_s)$ and assigned (signed) distance data $(\point_1,\sdistance_1),$ $\dots,$ $(\point_r,\sdistance_r)$, they begin by connecting the cell's dual vertex $\vertex_i$ to the Hermite points $\hpoint_1,\dots,\hpoint_s$ to build a \emph{local mesh} $\mesh_i(\vertex_i)$.
Then, they consider the error function
\begin{equation}\label{equ:energy}
\energy = \sum_{j=1}^r \left(\text{dist}(\point_j,\mesh_i(\vertex_i)) - |\sdistance_j|\right)^2 + \sum_{k=1}^s \left(\hnormal_k \cdot (\vertex_i - \hpoint_k)\right)^2\, ,
\end{equation}
in which the first term encourages the reconstruction to conform to the assigned (signed) distance data, while the second term encourages it to agree with the surface's normal direction at the Hermite points.
Using an iterative approach, they minimize this energy with respect to the position of the cell's dual vertex $\vertex_i$.
After this optimization is carried out for every active cell, the vertices $\vertex_i$ are joined together using each cell's active neighbors to form a quad mesh $\mesh$, which is in turn used to update the Hermite data.
This process is repeated through a number of local-global iterations, alternating between updating the reconstruction's vertices and improving the Hermite data, until a maximum number of iterations is reached. We refer readers interested in the details of this optimization to the original work by \citet{dcsdd}.

\subsection{Our algorithm: Contouring \emph{Unsigned} Distance Data}
\label{sec:algorithm}

The construction of our algorithm begins with the critical realization that the optimization described in the last two paragraphs of \autoref{sec:signedDC} is fundamentally agnostic to the sign in the distance data and the specific inward/outward orientation of the surface normals (note the absolute value around $\sdistance_j$ in \autoref{equ:energy} and the square after the dot product in the second term). 

\begin{figure}
\centering
\includegraphics[width=\linewidth]{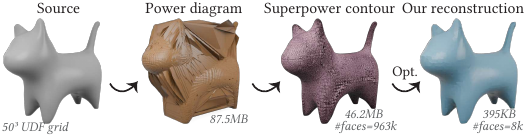}
\vspace{-0.6cm}
\caption{Our algorithm's pipeline: from the input unsigned distance data (left), we compute the power diagram (center left), which we filter to find the superpower contour (center right). It is in turn used to build an optimization problem, which we solve to obtain our final reconstruction (right).}
\vspace{-0.3cm}
\label{fig:pipeline}
\end{figure}

\begin{figure}
\centering
\includegraphics[width=\linewidth]{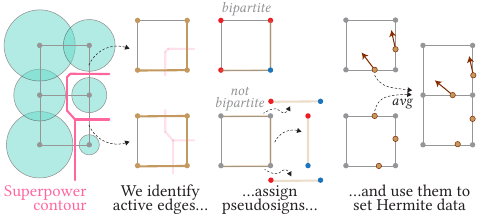}
\vspace{-0.6cm}
\caption{For each cell, we label as ``active'' any edges that intersect with the superpower contour (left). We use the graph defined by the active edges to assign per-cell (if bipartite) or per-edge (otherwise) pseudosigns, allowing us to initialize each edge's Hermite points and normals.}
\vspace{-0.3cm}
\label{fig:didactic-pseudosign}
\end{figure}

While promising, this critical realization does \emph{not} mean that the method of \citet{dcsdd} can be immediately applied to unsigned distance inputs. Indeed, it relies on changes in the distances' sign to identify the grid's active edges, which fundamentally determines the cells for which the optimization is carried out.
Additionally, for each active edge, it uses linear interpolation and finite differences on the grid nodes' signed distances to set the initial Hermite points and normals. This last step cannot be easily generalized to unsigned inputs, since the lack of sign causes linear interpolation to find no zero crossing and finite differencing yields poor quality estimates in the neighborhood of the non-differentiability caused by the surface's zero level set (see \autoref{fig:didactic-pseudosign}).

We now take advantage of the superpower contour $\spc$ introduced in \autoref{sec:method-theory}. In particular, its convergence properties (\autoref{thm:convergence}) will enable us to use $\spc$ as a proxy to localize the unknown surface $\surface$. Our algorithm thus begins by computing $\spc$ and testing it for intersections against all edges of a regular grid, labeling as active any edge with a non-empty intersection (see \autoref{fig:didactic-pseudosign}).

For each active cell, we will also use the superpower contour to assign to each node a local pseudo-sign that will enable us to build an initial Hermite point $\hpoint$ and normal $\hnormal$ on each of its active edges (see \autoref{fig:didactic-pseudosign})
Specifically, we interpret the cell's nodes as nodes in a graph, which are connected to one another if and only if they share an active edge (equivalently, if the cell's edge connecting them intersects the superpower contour $\spc$).

If this graph is bipartite (i.e., it contains no odd-length cycles, as in the top row of \autoref{fig:didactic-pseudosign}), the cell's nodes can be separated into two distinct, locally consistent sets such that every active edge crosses between them. In that case, the distance values in one of these sets can be made (temporarily) negative, and trilinear interpolation allows us to estimate the Hermite crossing point $\hpoint$ and its normal $\hnormal$ at every active edge.
If the graph is not bipartite (as in the bottom row of \autoref{fig:didactic-pseudosign}), all Hermite normals $\hnormal$ in the cell are set to zero, and pseudo-signs can be assigned on a per-active-edge basis to estimate $\hpoint$ using linear interpolation.

This process produces four different values of $\hpoint$ and $\hnormal$ at each active edge, corresponding to each of the cells containing it. As a final step of this initialization, we average these Hermite points and normals (flipping the latter as necessary to maintain a locally consistent orientation).
This leads to a configuration of Hermite data that one can set as the initialization to the optimization strategy detailed in \autoref{sec:signedDC}, resulting in a mesh $\mesh$ approximating $\surface$.
A full step-by-step pseudocode of our algorithm is available in \autoref{alg:pseudocode}.

\begin{figure}
\centering
\includegraphics[width=\linewidth]{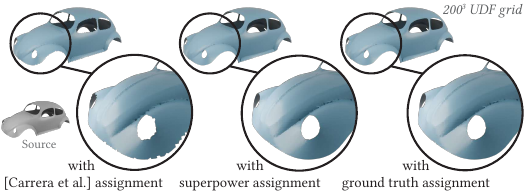}
\vspace{-0.6cm}
\caption{Our superpower contour's function as a proxy for the unknown geometry can be further exploited during optimization: we use it to assign distance data to different cells, producing boundary reconstruction similar to the ones one would obtain if using the ground truth geometry as proxy.}
\vspace{-0.3cm}
\label{fig:diff_optimize}
\end{figure}

\paragraph{Sphere assignment.}
To avoid an explosion in computational complexity, the first term in the energy in \autoref{equ:energy} loops only over a subset of the input distance data; namely, the data \emph{assigned} to the specific cell $\cell_i$. \citet{dcsdd} propose a simple strategy for this assignment: before each local optimization stage, one uses the previous iteration's reconstructed mesh $\mesh$ and finds $\closestpoint_1,\dots,\closestpoint_n$, the closest point on $\mesh$ to every one of the sample locations $\point_1,\dots,\point_n$. Then, each $\point_j$ is assigned to the grid cell containing its associated $\closestpoint_j$.
Because $\mesh$ is updated in every iteration, this data assignment must be repeated at every iteration of the optimization algorithm.

Experimentally, we find that this sphere assignment strategy often leads to erroneous reconstructions, especially near the boundaries of open surfaces.
Intuitively, a sub-optimal assignment can ``lock'' boundaries into zigzag configurations  in which the inner sections of the zigzag lack any assigned distance data (see \autoref{fig:diff_optimize}).
This is yet another algorithmic step in which we make use of our superpower contour $\spc$: since it serves as a reliable proxy of the unknown surface $\surface$, we can instead define $\closestpoint_j$ to be the closest point \emph{on $\spc$} to each sample location $\point_j$. 
Since $\spc$ is fixed, this choice removes the need to update the assignment during the optimization.
More critically, it avoids the ``locking'' caused by the traditional assignment and significantly improves our algorithm's reconstruction accuracy, especially for open shapes with boundaries.

\begin{figure}
\centering
\includegraphics[width=\linewidth]{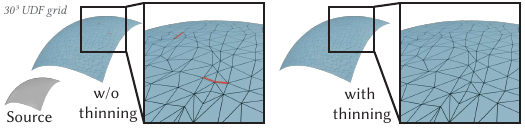}
\vspace{-0.6cm}
\caption{Our optional thinning step removes some of the small, undesirable non-manifold features common in dual contouring reconstruction methods.}
\vspace{-0.3cm}
\label{fig:ablation_thinning}
\end{figure}

\begin{figure}
\centering
\includegraphics[width=\linewidth]{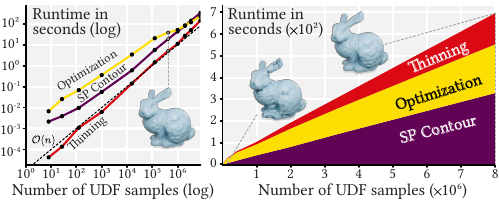}
\vspace{-0.6cm}
\caption{Our method's computational bottleneck consists of three steps, all of which scale log-linearly with the size of the input.}
\vspace{-0.3cm}
\label{fig:asymptotic_complexity}
\end{figure}

\paragraph{Thinning.}
Reconstructing a mesh from a discrete set of distance samples is a fundamentally underdetermined problem: there is an infinite number of surfaces that are consistent with the input distance data, and it is only through the use of priors that one can discriminate between them.
A common prior in discrete \emph{signed} distance reconstruction is topological: if the distances stored at a grid cell's corners exhibit a sign change, one assumes (e.g., through lookup tables \cite{marchingcube}) that $\surface$ is locally manifold and crosses the cell only once.

A similar prior is difficult to enforce directly on our unsigned formulation. Given the conservative nature of the superpower contour (see \autoref{thm:pc-in-spc}), the identified active edges will most often be a superset of the edges intersected by the true surface $\surface$, causing the final mesh to contain numerous undesirable non-manifold vertices and edges (see \autoref{fig:ablation_thinning}, highlighted).
Thus, instead, we will follow the lead of existing continuous UDF reconstruction methods \cite{wang2026voroudf} and rely on a post-processing \emph{thinning} of the final mesh that reduces topological noise.
In particular, we borrow the strategy proposed by \citet{wang2026voroudf}, in which the mesh is separated into its manifold connected components and sufficiently small boundary components (with cardinality less than $10$ in our experiments) are iteratively removed until none remain (see Algorithm \autoref{alg:pseudocode}).

\renewcommand{\graylinefrom}{34}
\renewcommand{\graylineto}{49}
{
\begin{algorithm}
\caption{\textbf{S}uper\textbf{p}ower Contouring of \textbf{U}nsigned \textbf{D}istance \textbf{D}ata. Lines in grey are adapted from \cite{dcsdd}.}\label{alg:pseudocode}
    \begin{algorithmic}[1]
    \Function{unsigned\_distance\_contouring}{$\point_1, \dots, \point_n,\distance_1, \dots, \distance_n$}
        \State $\{\pdf_{i,j}\} \gets$ compute power diagram of $\{(\point_i,\distance_i)\}$
        \State $\spc \gets \emptyset$
        \For{each power face $\pdf_{i,j}$}
            \If{$\pdf_{i,j} \cap \ball(\point_k,\distance_k) = \emptyset \ \forall k$}
                \State $\spc \gets \spc \cup \{\pdf_{i,j}\}$
            \EndIf
        \EndFor
        \For{each grid edge $\edge$}
            \If{$\edge \cap \spc \neq \emptyset$}
                \State Mark $\edge$ as active
            \EndIf
        \EndFor
        \For{each cell $\cell_i$ containing an active edge}
            \State $G_i \gets$ graph on cell corners with active edges
            \If{$G_i$ is bipartite}
                \State Assign pseudo-signs from bipartition of $G_i$
                \State Estimate $\Hermitepoint^0,\Hermitenormal^0$ on each active edge via 
                \Statex \hspace{\algorithmicindent}\hspace{\algorithmicindent}\hspace{\algorithmicindent}trilinear interpolation
            \Else
                \State $\Hermitenormal^0 \gets \mathbf{0}$ on each active edge
                \State Estimate $\Hermitepoint^0$ via linear interpolation
                \Statex \hspace{\algorithmicindent}\hspace{\algorithmicindent}\hspace{\algorithmicindent}along each active edge
            \EndIf
        \EndFor
        \For{each active edge $\edge$}
            \State $\Hermitepoint^0,\Hermitenormal^0 \gets$ average estimates over incident cells
        \EndFor
        \For{each active cell $\cell_i$}
            \State $\centroid_i \gets$ average of Hermite points in $\cell_i$
            \State $\cellvert_i^0 \gets \centroid_i$
        \EndFor
        \For{every data point $(\point_j,\distance_j)$}
            \State $\closestpoint_j \gets$ closest point to $\point_j$ on $\spc$
            \State Assign $(\point_j,\distance_j)$ to cell containing $\closestpoint_j$
        \EndFor
        \GrayFor{$k = 0,\dots,\text{max\_outer\_iter}$}
            \GrayFor{(parallel) every active cell $\cell_i$}
                \graystate{$\cellvert_i^{k,0} \gets \cellvert_i^{k}$}
                \GrayFor{$r=1,\dots,\text{max\_inner\_iters}$}
                    \graystate{$\cellvert_i^{k,r+1} \gets$ solve quadratic minimization}
                    \Statex \hspace{\algorithmicindent}\hspace{\algorithmicindent}\hspace{\algorithmicindent}\hspace{\algorithmicindent}\hspace{\algorithmicindent}\textcolor{algshade}{as proposed by \citet{dcsdd}}
                    \GrayIf{$\|\cellvert_i^{k,r+1} - \cellvert_i^{k,r}\| \leq \text{tol}$}
                        \State \textcolor{algshade}{\textbf{break}}
                    \GrayEndIf
                \GrayEndFor
                \graystate{$\cellvert_i^{k+1} \gets \cellvert_i^{k,r_{last}}$}
            \GrayEndFor
            \GrayFor{each active edge}
                \graystate{$\Hermitepoint^{k+1},\Hermitenormal^{k+1} \gets$ update via \citet{dcsdd}}
            \GrayEndFor
        \GrayEndFor
        \graystate{$\mesh \gets$ connect vertices around each active edge into quads}
        \While{true}
            \State Re-compute non-manifold components
            \While{boundary component with cardinality < $10$}
                \State Remove the component
            \EndWhile
            \State Break if only one component left
        \EndWhile
        \State \Return $\mesh$
    \EndFunction
    \end{algorithmic}
\end{algorithm}
}

%% file: sections/4_experiment.tex
\section{Experiments \& Results}
\label{sec:results}

We implemented our algorithm in C++ and Python, using CGAL's \cite{cgal} exact predicates to compute the sample data's power diagram and \emph{libigl} \cite{libigl} and \emph{Gpytoolbox} \cite{gpytoolbox} for common geometric subroutines.
We relied on \emph{Polyscope} \cite{polyscope} for visualizations during the algorithm's development, and rendered all our results with Blender and \emph{BlenderToolbox} \cite{blendertoolbox}.
All comparisons to prior work use their official open-source implementations unless explicitly specified otherwise; we sincerely thank all authors for releasing their code publicly.
Many of these implementations contain autodifferentiation components that require modern Nvidia GPUs; because of this, we conduct all timing experiments on the same Linux machine with an AMD EPYC 7763 64-Core CPU, eight NVIDIA RTX A6000 GPUs and 1 TB of memory.
Our method poses no such requirement on hardware, and all other results not containing timing information were generated on a MacBook Pro M3 with 36GB RAM.

\vspace{-0.1cm}
\subsection{Ablations and Experiments}
\label{sec:experiments}

Our algorithm's theoretical complexity is log-linear in the number of input unsigned distance samples, $n$.
In particular, as shown in \autoref{fig:asymptotic_complexity}, its runtime is split between the three major steps in our method.
First is the computation of the superpower contour $\spc$, which itself consists of calculating the power diagram $\pd$ (known to be $\bigO(n\log n)$ \cite{aurenhammer1987power}) and filtering the faces that intersect any of the spheres defined by the input data (also $\bigO(n\log n)$ with the help of an AABB tree).
While our theoretical results only bound the size of $\spc$ from below (see \autoref{thm:pc-in-spc} and \autoref{fig:pc_spc_2d_3d}), we do show that it converges to the power contour proposed by \citet{mnmeg} as the data size increases (see \autoref{fig:convergence_pc_spc}).
The log-linear theoretical cost of this entire algorithmic step is confirmed experimentally in \autoref{fig:asymptotic_complexity} (see purple line). 

Second is the local-global optimization described in \autoref{sec:algorithm}, with scaling  $\bigO(n k_1 k_2)$, that is, linear in the size of the input data ($n$) and the number of inner ($k_1$) and outer ($k_2$) optimization iterations.
While \citet{dcsdd} propose a large number of outer and inner iterations ($100$), this is partly due to its tendency for suboptimal assignments of spheres to active cells, which can only be resolved with repeated corrections.
As shown in \autoref{fig:didactic-pseudosign}, our method instead uses the superpower contour to assign distance data to active cells, significantly improving reconstruction accuracy at the cost of a single $\bigO(n\log n)$ (parallel) loop.
As an added benefit, this improved assignment strategy also causes our algorithm to converge significantly faster, requiring a much smaller number of inner and outer iterations (see \autoref{fig:outer-iters-ablation} and \autoref{fig:inner-iters-ablation}, and compare them to Figs. 6 and 16 in the paper by \citet{dcsdd}).
Unless stated otherwise, all our results use $k_1= 10$ and $k_2= 20$.

Finally, the last step of our algorithm is the thinning, in which one identifies the manifold components of the reconstruction and filters out those with small boundaries (see \autoref{fig:ablation_thinning}).
This is done with the help of a face-face adjacency matrix which is built and analyzed in log-linear time (see red line in \autoref{fig:asymptotic_complexity}), leading to our total algorithm's complexity $\bigO(n \log n)$.

We did not optimize our method beyond asymptotics; its wall-clock runtimes on a diverse set of examples can be seen in Figures \ref{fig:outer-iters-ablation}, \ref{fig:inner-iters-ablation}, \ref{fig:asymptotic_complexity} and \ref{fig:neural}.
In particular, \autoref{fig:runtime_comp} shows our method's runtime in the context of prior work in (continuous) UDF reconstruction.
As expected, it is generally slower than single-step spatial subdivision methods (top row, purple hues) and double-cover mesh evolution ones (bottom row, orange hues). 
Critically, our algorithm achieves these runtimes while avoiding the prominent artifacts and topological noise of prior methods when they are applied to discrete unsigned distance data, a fact that we explore further below.

\begin{figure}
\centering
\includegraphics{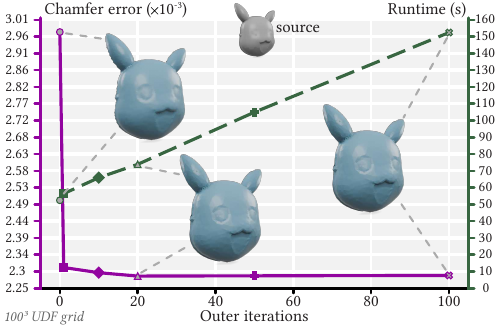}
\vspace{-0.3cm}
\caption{Our use of the superpower contour causes our optimization to converge with a much smaller number of iterations than the original work by \citet{dcsdd}.}
\vspace{-0.3cm}
\label{fig:outer-iters-ablation}
\end{figure}

\subsection{Comparisons to prior work}
\label{sec:comparisons}

To the best of our knowledge, no prior method considers the problem of reconstructing a mesh from a \emph{discrete} set of unsigned distance data.
Thus, we found no algorithm that directly accepts a lone set of UDF samples as input without assuming additional information is available.
Nonetheless, we did our best effort to adapt existing algorithms designed for similar or related tasks, comparing them to ours while making at every turn the experimental decisions most generous to them.

\paragraph{Comparison to \emph{continuous} UDF reconstruction}
Spatial subdivision methods like GeoUDF \cite{ren2023geoudf}, CapUDF \cite{zhou2022capudf} and MeshUDF \cite{guillard2022meshudf} assume knowledge of the UDF \emph{and its gradient} at each corner of a regular grid; thus, we adapted them to our task by estimating the gradients with central finite differences.
On the other hand, DualMeshUDF \cite{zhang2023dualmeshUDF} and double-cover methods like DCUDf \cite{hou2023dcudf} and DCUDF2 \cite{chen2025dcudf2} assume the ability to query both UDF and gradient at arbitrary positions: to test these, we use trilinear interpolation from the input UDF values and (estimated) gradients.

Experimentally, we found it impossible to reduce DualMeshUDF's reliance on querying the true UDF: in most of our tests with discrete distance data, its output was often an empty or almost-empty file, causing us to obviate it from all comparisons outside of \autoref{fig:baselines_simple_shape}.
Additionally, DCUDF and DCUDF2 end with a ``stitching'' or ``cut'' operation that joins the two surface layers into a single-layered output.
Experimentally, we found that this cut often removes large sections of the geometry; thus, for fairness, we report metrics with and without this step (see ``cut=T'' and ``cut=F'' in Tables \ref{tab:exp_tab_cad}, \ref{tab:exp_tab_deepfashion3d} and \ref{tab:sals}).

\begin{figure}
\centering
\includegraphics{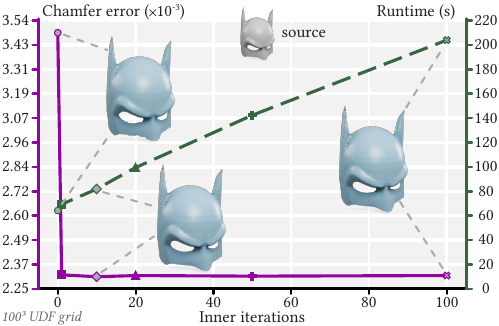}
\vspace{-0.3cm}
\caption{Very few inner optimization iterations are needed with our method, thanks to the quality proxy provided by the superpower contour}
\vspace{-0.3cm}
\label{fig:inner-iters-ablation}
\end{figure}

\begin{figure}
\centering
\includegraphics[width=\linewidth]{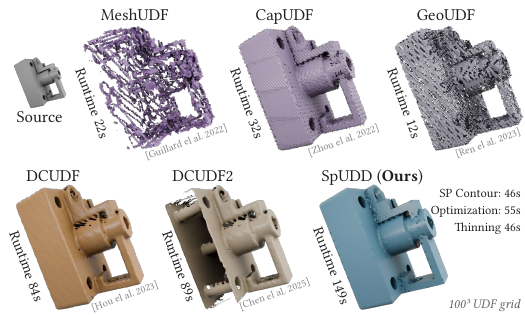}
\vspace{-0.6cm}
\caption{We avoid the artifacts of prior work at a modest runtime cost.}
\vspace{-0.3cm}
\label{fig:runtime_comp}
\end{figure}

Some of these methods are designed for general UDF inputs, while others make additional assumptions on the geometry (\eg., those based on marching cubes cannot output non-manifold surfaces).
Thus, in the same interest of fairness and generosity towards prior work, we separate our quantitative evaluation into three separate datasets: a commonly-used randomly selected subset of the ABC dataset, made up exclusively of closed shapes \cite{koch2019abc}; a dataset of open surfaces corresponding to garments \cite{zhu2020deepfashion}; and a recently introduced dataset of intersecting shapes with non-manifold features \cite{ren2025sals}. For each shape, we sample a $100^3$ grid of unsigned distance values, and measure Hausdorff, Chamfer and Edge-Chamfer distances between each method's geometry and that of the ground-truth.

Quantitative results for each dataset are presented in Tables \ref{tab:exp_tab_cad}, \ref{tab:exp_tab_deepfashion3d} and \ref{tab:sals}. Our algorithm, tailor-made for the discrete case, outperforms all others, often by several integer factors, across all metrics.
The reasons for this are shown qualitatively in Figures \ref{fig:teaser}, \ref{fig:baselines_simple_shape}, \ref{fig:runtime_comp}, \ref{fig:dataset_ABC}, \ref{fig:dataset_deepFashion3D} and \ref{fig:dataset_SALS}: when forced to rely only on discrete data, continuous UDF reconstruction methods produce visible artifacts related to grid alignment, and often miss entire geometric regions or produce almost no output.

\begin{table}
\caption{Quantitative comparison on the 100 ABC models~\cite{xu2024cwf} against seven state-of-the-art methods, using resolution $100^3$. We report the number of vertices (\#V) in the reconstructed mesh, the $L_2$ chamfer error ($\cd$), Hausdorff distance error ($\hd$) and chamfer edge error ($\ecd$) in the table. \best{Best} scores are shown in bold, the \second{second best} scores are underlined.
}
\vspace{-0.4cm}
\begin{center}
\resizebox{\columnwidth}{!}{%
\rowcolors{2}{gray!10}{white}
\begin{tabular}{c||c|c|c|c}
\rowcolor{header2color}
\hline
\toprule
 Method & $\cd$ x$10^{-3}$~$\downarrow$ & $\hd$ x$10^{-2}$~$\downarrow$  & $\ecd$ x$10^{-2}$~$\downarrow$ & \#Vertices\\
\midrule
MeshUDF
    &73.03 &25.09 &6.55 &11k  \\ 
CapUDF
    &6.34 &4.28 &3.97 &138k  \\ 
GeoUDF
    &11.91 &4.65 &7.60 &24k \\ 
    \midrule
DCUDF (cut=F)
    &12.98 &3.73 &8.46 &97k \\ 
DCUDF (cut=T)
    &21.02 &5.86 &10.61 &54k  \\ 
DCUDF2 (cut=F)
    &7.03 &4.53 &13.08 &103k  \\ 
DCUDF2 (cut=T)
    &8.77 &5.70 &15.34 &57k  \\ 
    \midrule
Ours w/o thinning
    &\second{3.12} &\second{1.44} &\second{1.38} &23k  \\ 
Ours
    &\best{3.06} &\best{1.39} &\best{1.27} &23k  \\ 
\bottomrule
\end{tabular}}
\end{center}
\label{tab:exp_tab_cad}
\vspace{-0.3cm}
\end{table}

\paragraph{Comparison to discrete \emph{SDF} reconstruction}
Most algorithms to reconstruct meshes from discrete signed distance  data with publicly available implementations cannot be easily adapted to generic unsigned inputs.
For example, \emph{Reach for the Spheres} \cite{rfts} assumes prior knowledge of the surface's topology, and does not support open or non-manifold surfaces (see \autoref{fig:rfts}).
\emph{Reach for the Arcs} \cite{rfta} and the follow-up by \citet{mnmsiggraph} rely on Poisson Surface Reconstruction \cite{Kazhdan2006} as a fundamental building block, which requires a point cloud whose normal orientations are determined by the data's sign.

At the same time, our algorithm builds on the insights made by \citet{mnmeg}, who define the \emph{power contour} $\pc$ of a set of signed distance samples. We use similar intuitions to construct the \emph{superpower contour} $\spc$, which is guaranteed to contain $\pc$ and, like it, converges to the true surface as sampling density increases (see \autoref{fig:convergence_pc_spc}).
Because we draw inspiration from this work, we thought it appropriate to consider how one may combine their algorithm with our proposed superpower contour $\spc$ to produce reconstructions from unsigned distance data.
For example, one could follow their algorithm to construct the distance data's regular tetrahedralization, and then test each tet edge for intersections against $\spc$ to mark them as \emph{active} as in our \autoref{sec:algorithm}.
Given this set of active edges, one could follow a similar logic to our Hermite data initialization in \autoref{sec:algorithm}, building a local graph at each tet and assigning a pseudo-sign to its vertices if it is bipartite.
A surface patch can be produced for each such tet using Marching Tetrahedra \cite{doi1991efficient}, while tets that are not bipartite may be left empty.

We refer to the above strategy as ``RC$^*$'', and show its performance in \autoref{fig:baselines_simple_shape} and \autoref{fig:mm3_ours} (left-most meshes in yellow hues).
Unfortunately, as both figures show, the preponderance of non-bipartite tetrahedra causes meshes produces by RC$^*$ to contain large holes and topological noise, motivating the development of the method described in \autoref{sec:method-practice} (see \autoref{fig:baselines_simple_shape} and \autoref{fig:mm3_ours}, in blue).

\begin{table}
\caption{Quantitative comparison on the first 100 DeepFashion3D models~\cite{williams2019deep} against seven state-of-the-art methods, using resolution $100^3$. We report number of vertices (\#V), $L_2$ chamfer error ($\cd$), Hausdorff  error ($\hd$) and chamfer edge error ($\ecd$) in the table. \best{Best} scores are shown in bold, the \second{second best} scores are underlined.
}
\vspace{-0.4cm}
\begin{center}
{
\setlength{\tabcolsep}{12pt}%
\resizebox{\columnwidth}{!}{%
\rowcolors{2}{gray!10}{white}
\begin{tabular}{c||c|c|c}
\rowcolor{header2color}
\hline
\toprule
 Method  & $\cd$ x$10^{-3}$~$\downarrow$ & $\hd$ x$10^{-2}$~$\downarrow$ &\#Vertices \\
\midrule
MeshUDF
    &56.91 &23.71 &12k  \\ 
CapUDF
    &6.28 &3.85 &128k  \\ 
GeoUDF
    &11.39 &5.29 & 24k \\ 
    \midrule
DCUDF (cut=F)
    &12.88 &4.23 &94k   \\ 
DCUDF (cut=T)
    &23.03 &10.89 &50k   \\ 
DCUDF2 (cut=F)
    &6.24 &4.87 &98k   \\ 
DCUDF2 (cut=T)
    &10.66 &9.34 &53k   \\ 
    \midrule
Ours w/o thinning
    &\second{3.68} &\best{1.90} &23k  \\ 
Ours
    &\best{3.65} &\second{2.37} &22k  \\ 
\bottomrule
\end{tabular}}
}
\end{center}
\label{tab:exp_tab_deepfashion3d}
\end{table}

\begin{figure}
\centering
\includegraphics{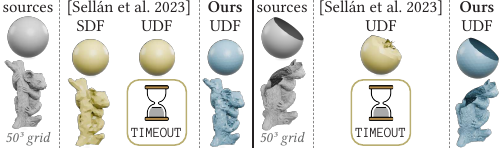}
\vspace{-0.7cm}
\caption{\emph{Reach for the Spheres} \cite{rfts} can reconstruct the simplest closed shapes from unsigned distance inputs, but fails for complex geometries and shapes with non-zero genus.}
\vspace{-0.3cm}
\label{fig:rfts}
\end{figure}

\begin{figure}
\centering
\includegraphics{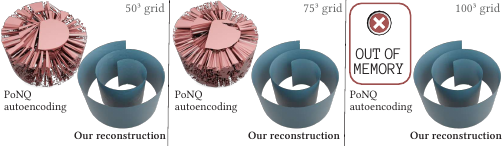}
\vspace{-0.7cm}
\caption{Executing the authors' provided implementation for open surfaces caused PoNQ \cite{maruani2024ponq} to fail for a simple open shape, which our method accurately reconstructs.}
\vspace{-0.3cm}
\label{fig:ponq_comp}
\end{figure}

\paragraph{Comparison to neural reconstruction methods}
A significant amount of recent work has proposed resolving the ambiguities inherent to unsigned distance data through \emph{data-driven, neural} algorithms that learn priors from large datasets of shapes.
Some of these are nonetheless unsuited for discrete UDF inputs: for example, NSUDF \cite{stella2024nsdudf} assumes the ability to sample the continuous UDF and produces clear grid artifacts if forced to interpolate values from a grid (see \autoref{fig:neural}, left).
Similarly, PoNQ \cite{maruani2024ponq} is only able to reconstruct open surfaces by closing them, reconstructing them as solids, and adding holes in post-processing, a strategy that is invalid for most UDF inputs, even without discretization (see \autoref{fig:ponq_comp}, produced using the authors' official implementation script for open surfaces).

Of all data-driven methods that we tested, only \emph{Neural Dual Contouring} (NDC) \cite{Chen2022} was able to reconstruct arbitrary surfaces from discrete unsigned distance inputs.
As shown in \autoref{fig:neural}, our method matches or slightly outperforms the author-provided pretrained NDC on representative examples from all three datasets.
Critically, while \citet{rfts} showed one can outperform NDC by increasing the resolution of the output, our method matches NDC's data-driven prior \emph{at the same resolution}, since both rely on a similar contouring strategy which produces a single vertex per active cell.

Coupled with NDC's impressive generalization ability, this finding may prove relevant in the development of future neural approaches. Unlike in other realms, it is possible that data-driven strategies for SDF and UDF reconstruction are being used not to disambiguate between solutions in an underdetermined setting but to encode the solution to a complex optimization problem. 

\begin{table}
\vspace{-0.4cm}
\caption{Quantitative comparison on 33 self-intersecting models~\cite{ren2025sals} against seven state-of-the-art methods, using a same resolutions $100^3$. We report number of vertices, $L_2$ chamfer error ($\cd$), Hausdorff distance error ($\hd$) and chamfer edge error ($\ecd$) in the table. \best{Best} scores are shown in bold, the \second{second best} scores are underlined.}
\vspace{-0.4cm}
\begin{center}
\resizebox{\columnwidth}{!}{%
\rowcolors{2}{gray!10}{white}
\begin{tabular}{c||c|c|c|c}
\rowcolor{header2color}
\hline
\toprule
 Method & $\cd$ x$10^{-3}$~$\downarrow$ & $\hd$ x$10^{-2}$~$\downarrow$  & $\ecd$ x$10^{-2}$~$\downarrow$ &\#Vertices \\
\midrule
MeshUDF
    &44.91 &19.71 &3.51 &26k   \\ 
CapUDF
    &9.50 &4.82 &2.89 &209k  \\ 
GeoUDF
    &12.99 &5.61 &3.81  &39k \\ 
    \midrule
DCUDF (cut=F)
    &14.49 &4.95 &5.13 &140k  \\ 
DCUDF2 (cut=F)
    &9.34 &5.77 &5.46 &145k  \\ 
    \midrule
Ours w/o thinning 
    &\second{6.02} &\best{1.96} & \second{2.16} &41k \\
Ours
    &\best{5.97} &\second{2.18} &\best{2.14} &39k  \\ 
\bottomrule
\end{tabular}}
\end{center}
\label{tab:sals}
\vspace{-0.3cm}
\end{table}

\vspace{-0.1cm}
\subsection{Applications}
\label{sec:applications}
Discrete implicit representations have seen a large amount of use in applications from physical simulation to 3D generative and artist-driven modeling.
However, while algorithms like Marching Cubes \cite{marchingcube,wyvill1986data} have enabled one to extract meshes from discrete \emph{signed} representations, no similar strategy existed for discrete \emph{unsigned} distance data, limiting the applications of this data format to closed shapes for which an outside/inside segmentation is well defined.
Our method thus broadens the space of shapes one can utilize in these applications; as shown in \autoref{fig:close_open_nonmanifold}, it allows one to model closed shapes as easily as open and non-manifold ones containing intersections.

Most implicit modeling frameworks involve the construction of elaborate CSG trees in which shapes are combined with each other through Boolean operations like unions.
Existing software used by both artists and engineers often relies on discrete representations to store, combine and reconstruct shapes efficiently.
As shown in Figure~\ref{fig:modeling_moon}, our method can be used to extend discrete CSG trees to open, non-manifold and intersecting surfaces.

%% file: sections/5_conclusion.tex
\vspace{-0.1cm}
\section{Limitations and Conclusion}
\label{sec:conclusion}
We have proposed the first algorithm tailor-made to reconstruct polygonal meshes from discrete unsigned distance samples. To construct it, we introduced a novel structure, called the \emph{superpower contour}, and mathematically proved its relation to prior work and convergence to the true, unknown surface.
Through exhaustive quantitative and qualitative tests, we have shown the superiority of our method against most conceivable alternatives.

Our algorithm's implementation achieves this performance in part by assuming that the distance data is stored at nodes of a regular grid containing the surface. We are hopeful for future work that uses our superpower contour to guide the refinement of an adaptive spatial subdivision that could be used to efficiently reconstruct meshes from distances sampled at any spatial positions.

\begin{figure}
\vspace{-0.4cm}
\centering
\includegraphics[width=\linewidth]{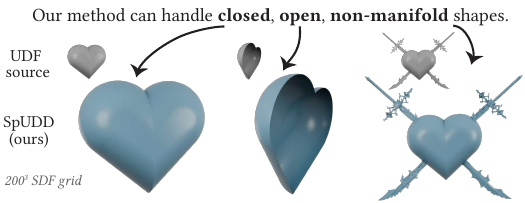}
\vspace{-0.7cm}
\caption{By broadening the shapes that can be reconstructed from discrete distance data, we enable a new set of geometric modeling applications.}
\vspace{-0.3cm}
\label{fig:close_open_nonmanifold}
\end{figure}

\begin{figure}
\centering
\includegraphics[width=\linewidth]{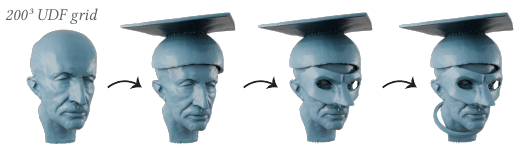}
\vspace{-0.6cm}
\vspace{0.2cm}
\includegraphics[width=\linewidth]{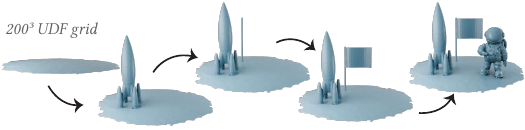}
\vspace{-0.6cm}
\caption{\emph{A giant leap} in artists' ability to iteratively model complex scenes involving open and non-manifold structures.}
\vspace{-0.3cm}
\label{fig:modeling_moon}
\end{figure}

\begin{figure*}
\centering
\includegraphics[width=\linewidth]{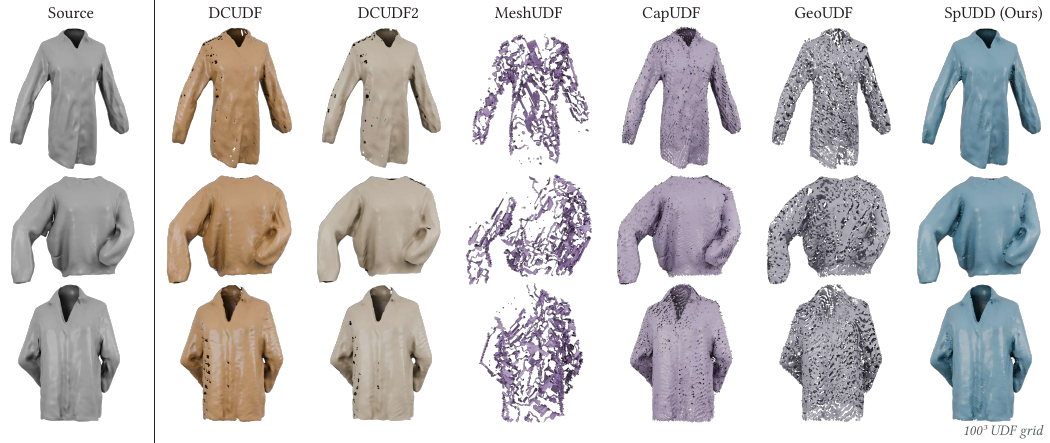}
\vspace{-0.7cm}
\caption{Visual comparison with 5 state-of-the-art methods using $100^3$ UDF grid on the DeepFashion3D dataset~\cite{zhu2020deepfashion}.}
\label{fig:dataset_deepFashion3D}
\end{figure*}

To prove the theoretical results related to the superpower contour and its limiting behavior, we assumed that the input data $\distance_i$ are the exact distances from each sample point $\point_i$ to the unknown surface $\surface$.
Our results will not hold if $\distance_i$ are instead assumed to belong to a wider class of implicit representations, like generic unit gradient fields or 1-Lipschitz quasi-distances.
For similar reasons, our method is only resilient to small amounts of noise in the distance data, quickly exhibiting artifacts if its magnitude increases  (see \autoref{fig:noise}).
Generalizing our algorithm beyond exact distance data will likely require fundamental modifications, from constructing a noise-aware equivalent to the power diagram $\pd$ to including noise tolerances in the local-global optimization described in \autoref{sec:algorithm}.

We hope to use this paper to introduce the specific problem of discrete unsigned distance reconstruction to the geometry processing community.
The main challenge lies in its underdetermined nature, as an infinite number of arbitrarily complex surfaces satisfy any coarse set of discrete distance samples.
This often leads our algorithm to produce surfaces that comply with the input distance data but contain dents near sharp and non-manifold features, as well as small topological features (see, e.g., the right-most column in \autoref{fig:dataset_SALS}).
Future work may consider novel ways of imposing common reconstruction priors like topological simplicity or smoothness in the absence of the inside-outside segmentation provided by signed data or the ability to obtain additional distance queries.

Due to their simplicity, efficiency, and flexibility, discrete SDFs have been used pervasively in applications from level set methods to implicit modeling.
However, while comparably easier to manipulate and able to represent a wider space of geometries, the use of discrete UDFs has been limited by the inability to robustly recover an explicit representation of the surface.
We hope that our work can serve to alleviate this limitation, and thus open the door for others to explore the myriad applications of discrete unsigned distance data in geometry processing and beyond.

\begin{figure}
\centering
\includegraphics{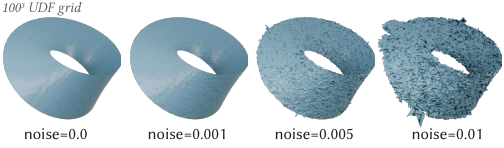}
\vspace{-0.6cm}
\caption{Our algorithm is moderately robust to small amounts of noise, but becomes unreliable as its magnitude increases.}
\vspace{-0.3cm}
\label{fig:noise}
\end{figure}

\begin{figure*}
\centering
\includegraphics[width=\linewidth]{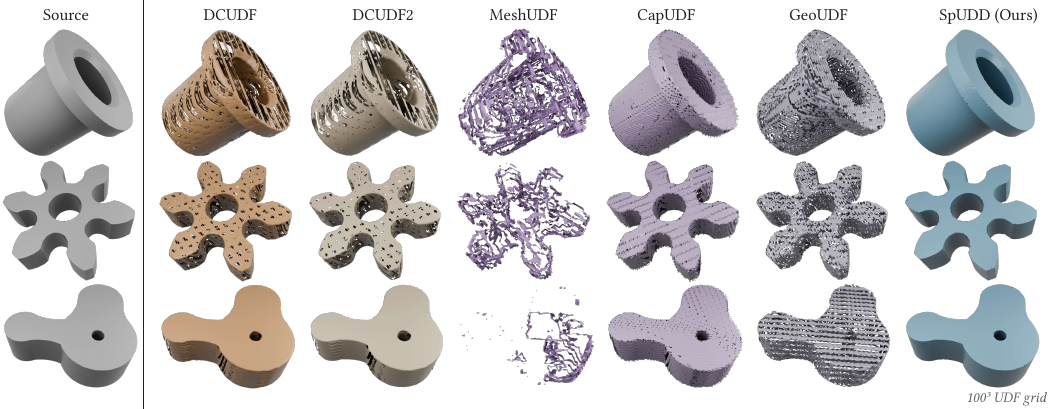}
\caption{Visual comparison with 5 state-of-the-art methods using $100^3$ UDF grid on the ABC dataset~\cite{koch2019abc}.}
\vspace{-0.3cm}
\label{fig:dataset_ABC}
\end{figure*}